\documentclass[11pt]{article}

\usepackage[dvipsnames]{xcolor}
\usepackage{geometry}
\usepackage{setspace} 
\usepackage{amssymb,amsmath,amsthm}
\usepackage[utf8]{inputenc}
\usepackage{amsthm}
\usepackage{natbib}
\usepackage{float}
\usepackage{nicematrix}
\usepackage{resizegather}
\usepackage{bbding}
\usepackage{tikz}
\usepackage{dsfont}
\usepackage{csquotes}
\usepackage{graphicx}
\usepackage{subcaption}
\usepackage{mathtools}

\mathtoolsset{showonlyrefs=true}

\definecolor{cornellred}{rgb}{0.7, 0.11, 0.11}

\usepackage[colorlinks=true, citecolor=cornellred, linkcolor=black]{hyperref}

\newtheorem{theorem}{Theorem}[section]
\newtheorem{proposition}[theorem]{Proposition}

\newtheorem{lemma}[theorem]{Lemma}
\newtheorem{definition}[theorem]{Definition}
\theoremstyle{remark}
\newtheorem{remark}{Remark}[section]

\title{A Generalization of Lanchester's Model of Warfare}
\author{N. Cangiotti{$^\star$}, M. Capolli{$^\ast$} and M. Sensi{$^\dagger$}\\[1em]
\small $^\star$Polytechnic University of Milan, Department of Mathematics, via Bonardi 9, \\ \small 20133 Milan, Italy. Email: \texttt{nicolo.cangiotti@polimi.it}\\
 \small $^\ast$Polish Academy of Sciences, Institute of Mathematics, Jana i Jedrzeja Sniadeckich 8,\\ \small 00-656 Warsaw, Poland. Email: \texttt{mcapolli@impan.pl}
\\
\small $^\dagger$MathNeuro Team, Inria at Universit\' e C\^ote d'Azur, 2004 Rte des Lucioles,\\ \small 06410 Biot, France. Email: \texttt{mattia.sensi@inria.fr}}

\date{}

\begin{document}

\maketitle
\begin{abstract} 
The classical Lanchester's model is shortly reviewed and analysed, with particular attention to the critical issues that intrinsically arise from the mathematical formalization of the problem. We then generalize a particular version of such a model describing the dynamics of warfare when three or more armies are involved in the conflict. Several numerical simulations are provided.
\bigskip


\smallskip

\end{abstract}
\section{Introduction}

Human interactions supply many phenomena, which can be modelled and analyzed with applied mathematics. One of the most famous models, developed in this context, concerns the role of the military strategy (and the linked decision problems) during a conflict between two or more armies. It is strange to think that, despite the numerous wars from ancient time up to nowadays, a rigorous mathematical modeling of warfare strategies is only fairly recent. 
In fact, the construction of a consistent theory reconciling the intuitive idea with the mathematical formalism only came close to World War I. In 1916, Frederick \cite{Lanchester1916,Lanchester1956} proposed a fist system of differential equations to model the clash between two armies. 

As we shall explain extensively in the next section, Lanchester provided a useful characterization of the dynamics of the battle by using two simple differential equations. Starting from these equations, he proved an interesting result, the so called \emph{Lanchester's square law}, which provides important information concerning the numerical evolution of the two forces. In the following decades, Lanchester's law was studied by many authors as 
\cite{Bracken95,Davis1995,Fox2010,Taylor1983}. Over the years, various generalization were proposed, starting from the action of mixed forces due to \cite{Lepingwell1987}. Most of these attempts are summarized by MacKay in \cite{MacKay2006}. Moreover, the applications of Lanchester's type models goes beyond the original purpose, as its core ideas were used to describe human, by \cite{Johnson2015}, and nonhuman, by \cite{Clifton2020}, evolution. Recently, the interest of such a model made a comeback in the operational research literature as series of works proposing new stimulating models. In particular, a game theoretical approach is given in \cite{Kress2018a, Kress2018b}. This approach exploits a matrix formalism similar to the one used in this work, as well as a network-theoretical approach used in \cite{kalloniatis2021optimising}.

The asymptotic dynamics were thoroughly analyzed for two heterogeneous (mixed) forces by \cite{MacKay2009} and by \cite{Lin2014}, in order to determine the optimal fire distribution policy, and for models that capture irregular warfare, such as insurgencies studied by \cite{Kress2020}. The authors \cite{Liu2012} moved some criticism to \cite{MacKay2009}, to which MacKay replied stressing the strengths of his approach. Lastly, we mention another sophisticated attempt which involves the use of partial differential equations with time variables developed by \cite{Sprandlin2007}.

Given the significance of such a model, we believe it is worth providing an accurate formal analysis
and a precise model in terms of dynamical systems research. Thus, the aim of this paper is to provide a mathematically consistent generalization of the classical model that sorts out the critical issues arising from the original construction. This construction allows us to describe general higher dimensional scenarios.

The paper is organized as follows. In Section \ref{Sec2}, we summarize in a mathematical framework the classical Lanchester's model and its related quadratic and linear laws. Section \ref{Sec3} contains our main contribution to the modelling effort. It is devoted to the proposal of Lanchester's type models with more than two armies. The attempt here is to provide a more mathematically coherent description of the dynamics of the battles, which is consistent with a realistic interpretation. As we shall see, there are various possible combinations, which lead to different interpretations of the real-world scenarios we aim to describe. Our analysis moves between these different scenarios, supplying various results which are validate with numerical simulations. Finally, in Section \ref{Sec4}, we take stock of the work, arguing possible improvements of our models and discussing leads for future developments.

\section{Classical Lanchester's Model}
\label{Sec2}

We begin by reviewing the well-known \emph{equations of warfare}, first formulated by \cite{Lanchester1916}, and then extensively developed by \cite{Taylor1983,MacKay2006,MacKay2009,Kress2018a,Kress2018b}. The aim of these equations was to model the change in combat power of two enemy factions facing each other. Throughout the paper, we indicate with $F'$ the derivative of the function $F$ with respect to the time variable $t$. We do not normalize the size of the armies, as it brings no advantage to our analysis, and we pick their values at time $t=0$ in order to produce clearly interpretable figures.

Lanchester's model can be summarized by two systems of differential equations, one for the \emph{aimed} fire:
\begin{equation}\label{lanch_aim}
\begin{cases}
\displaystyle{R'=-gG},\vspace{10pt}\\
\displaystyle{G'=-rR},
\end{cases}
\end{equation}
and one for the \emph{unaimed} one:
\begin{equation}\label{lanch_unaim}
\begin{cases}
\displaystyle{R'=-gRG},\vspace{10pt}\\
\displaystyle{G'=-rRG}.
\end{cases}
\end{equation}

Here, $G=G(t)$ and $R=R(t)$ represent the size of the two factions (Green army and Red army) at any time $t>0$; whereas $g,r>0$ are coefficients which indicate the fighting ability of each unit in the corresponding faction. The analysis of these systems provide the celebrated Lanchester's quadratic and linear law respectively.
\subsection{Lanchester's Quadratic Law}

We start  by providing a classical result for system \eqref{lanch_aim}.
\begin{theorem}[Lanchester's Quadratic Law]
\label{Thm:LanchSquare}
The solutions to \eqref{lanch_aim} satisfy the state equation
\begin{equation}\label{quad_law}
gG^2(t)-rR^2(t) = K,
\end{equation}
where $K$ is a constant that depends only on the fighting abilities of the two armies, $g$ and $r$, and on their initial size $G(0)$ and $R(0)$. 
\end{theorem}
The claim follows by derivating \eqref{quad_law} with respect to the time variable $t$ along solutions of \eqref{lanch_aim}.
Equation \eqref{quad_law} allows an \emph{a priori} estimate of the winning side: if $C>0$, then only $R(t)$ can ever be zero and thus faction G will win. 

It is important to notice the presence of squares in \eqref{lanch_aim}. It indicates that the outcome of the battle is more sensitive to the changes in the number of units rather than to changes in their effectiveness.

While equation \eqref{quad_law} is interesting for practical application, from a mathematical point of view we may be interested in solving system \eqref{lanch_aim} explicitly.
\begin{lemma}
The solutions of \eqref{lanch_aim} are:
\begin{equation}\label{solution_aim2D}
    \begin{cases}
    G(t)= \displaystyle{G(0) \cdot  \cosh \left(\sqrt{gr} \cdot  t\right)-R(0)\sqrt{\frac{r}{g}} \sinh \left(\sqrt{gr} \cdot  t\right)},\\
    \\
    R(t)= \displaystyle{R(0) \cdot \cosh \left(\sqrt{gr} \cdot  t\right)-G(0)
    \sqrt{\frac{g}{r}} \sinh \left(\sqrt{gr} \cdot  t\right)}.
    \end{cases}
\end{equation}
In particular, in the case $rR^2(0) = gG^2(0)$, the solutions to \eqref{lanch_aim} are reduced to:
\begin{equation}
    \begin{cases}
    G(t)=G(0)\exp(-\sqrt{gr}t), \\
    R(t)=R(0)\exp(-\sqrt{gr}t).
    \end{cases}
\end{equation}
\end{lemma}

Notice that, since $\cosh$ and $\sinh$ are asymptotically equivalent, then the behaviour of system \eqref{solution_aim2D} when $t\to+\infty$ is completely determined by the initial constants $G(0),g,R(0)$, and $r$. In particular whenever $gG(0)^2 \neq rR(0)^2$ one of the two will diverge to $+\infty$, while the other will diverge to $-\infty$, see Figure \ref{fig:quadr}.

\subsection{Lanchester's Linear Law}

For the case of the linear law described in system \eqref{lanch_unaim}, a result similar to Theorem \ref{Thm:LanchSquare} holds.

\begin{theorem}[Lanchester's Linear Law]
\label{Thm:LanchLinear}
The solutions to \eqref{lanch_unaim} satisfy the state equation
\begin{equation}\label{lin_law}
gG(t)-rR(t) = K,
\end{equation}
where $K$ is a constant that depends only on the fighting abilities of the two armies, $g$ and $r$, and on their initial size $G(0)$ and $R(0)$. 
\end{theorem}

As for Theorem \ref{Thm:LanchSquare}, the proof of this theorem can be obtained by direct derivation with respect to the time variables along solutions of \eqref{lanch_unaim}. The interpretation of equation \eqref{lin_law} is fundamentally different from the one of equation \eqref{quad_law}. The linear law does not gives a predominant role to the size of the army with respect to his combat power. In the case of the square law, a lack of combat expertise could be easily balanced by rising the number of soldiers by a small amount (i.e., increasing $G(0)$ or $R(0)$, respectively), whereas in the case of the linear law this is no longer true.
\smallskip

As for the case of the quadratic law, we may be interested in the mathematical solution of \eqref{lanch_unaim}.
\begin{lemma}\label{expo_dec}
Let us suppose $rR(0) \neq gG(0)$, then the solutions to \eqref{lanch_unaim} are:
\begin{equation}\label{sol_unaim}
    \begin{cases}
    G(t)=\displaystyle{\frac{(rG(0)R(0)-gG(0)^2)\exp{((gG(0)-rR(0))t)}}{rR(0) - gG(0)\exp{((gG(0)-rR(0))t)}}}, \\
    \\
    R(t)=\displaystyle{\frac{rR(0)^2-gG(0)R(0)}{rR(0) - gG(0)\exp{((gG(0)-rR(0))t)}}}.
    \end{cases}
\end{equation}
Moreover, in the case $rR(0) = gG(0)$, the solutions to \eqref{lanch_unaim} are reduced to:
\begin{equation}
    \begin{cases}
    G(t)=\displaystyle{\frac{G(0)}{1+gG(0)t}}, \\
    \\
    R(t)=\displaystyle{\frac{R(0)}{1+rR(0)t}}.
    \end{cases}
\end{equation}
\end{lemma}

Notice that, in contrast with the solutions of \eqref{lanch_aim}, one of the solutions of \eqref{lanch_unaim} tends to a positive limit, while the other one tends to 0 as $t\to\infty$, see also Figure \ref{fig:line}.

\begin{figure}[ht]  
\centering
    \begin{subfigure}[t]{0.45\textwidth}
        \centering
    \includegraphics[width=0.99\textwidth]{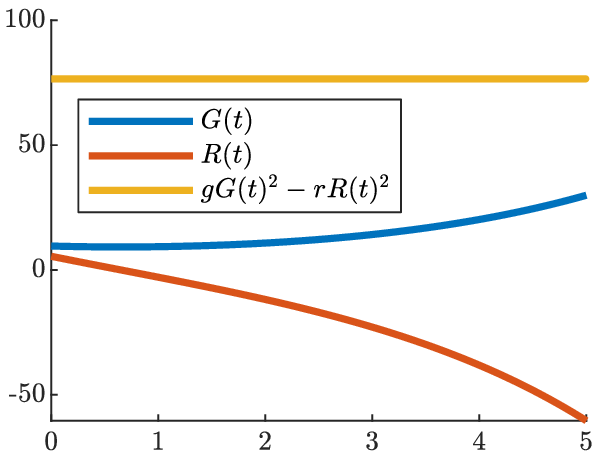}
    \caption{Typical solution of system \eqref{lanch_aim}.}
    \label{fig:quadr}
     \end{subfigure}
    \hspace{0.05\textwidth}
         \begin{subfigure}[t]{0.45\textwidth}
        \centering
         \centering
    \includegraphics[width=0.99\textwidth] {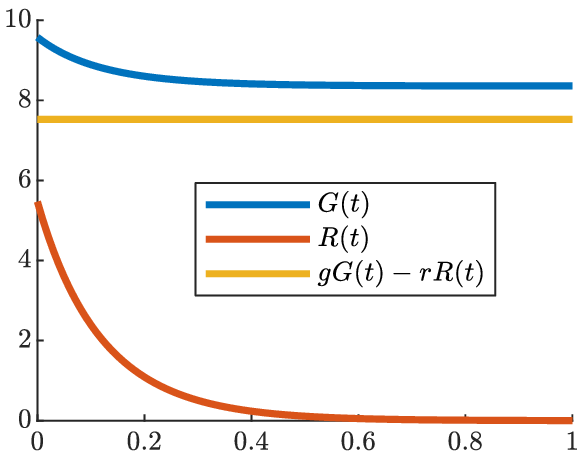}
    \caption{Typical solution of system \eqref{lanch_unaim}.}
    \label{fig:line}
         \end{subfigure}
         \caption{Evolution in time of (a) solutions of \eqref{lanch_aim}, plotted together with \eqref{quad_law}, with $G(0)$ and $R(0)$ chosen randomly in $(0,10)$ and $r=0.2$, $g=0.9$; and (b) solutions of \eqref{lanch_unaim}, as well as the constant of motion \eqref{lin_law}, with $G(0)$ and $R(0)$ as in Figure \ref{fig:quadr} and $r=0.2$, $g=0.9$. Notice that in (a) $R(t)$, the size of the losing army, quickly becomes negative and keeps decreasing. In turn, the negative sign of $R$ makes $G$ increase and asymptotically diverge to $+\infty$. Instead, in (b) $R(t)$, the size of the losing army, quickly approaches 0 as the system asymptotically converges to its equilibrium $G_{\infty}=G(0)-rR(0)/g$, $R_{\infty}=0$.}
         \label{figure_1}
\end{figure}

\section{Multi-Battle Model}
\label{Sec3}

In this section, we provide multi-dimensional generalizations of systems \eqref{lanch_aim} and \eqref{lanch_unaim}. Such a generalization can model a multilateral warfare between $n$ armies, if we assume all the attacks are carried out simultaneously. For ease of notation, in the following we shall use $X_1, X_2, \dots, X_n$ to represent the size of the armies involved.

Consider a column vector $X=(X_1, X_2, \dots, X_n)^T$, representing all the $n$ armies involved; assume each army has a fighting ability $x_i$, $i=1,2,\dots,n$, that we collect in the vector $x=(x_1,x_2,\dots,x_n)$.

\begin{definition}\label{def_conflictmatrix}
We define the \emph{conflict matrix} $\tilde{C}$ as the matrix whose entries are
$$
\tilde{c}_{i,j}=\begin{cases}
1 & \text{if army }X_i\text{ and }X_j\text{ are in conflict},\\
0 & \text{otherwise}.
\end{cases}
$$
The associated \emph{weighted conflict matrix} $C$ is obtained by re-scaling $\tilde{C}$, normalizing the column sum to 1. Its entries are
\begin{equation}\label{matrixC}
c_{i,j}=\dfrac{\tilde{c}_{i,j}}{\sum_{k=1}^n \tilde{c}_{k,j}},
\end{equation}
where the re-scaling takes into account the homogeneous division of forces of each army into the various conflicts it is fighting.
We will discuss other re-scaling options in Section \ref{Sec4}.

\end{definition}

Notice that the matrix $\tilde{C}$ is always symmetric\footnote{Here we are assuming that the conflicts are bilateral in the sense that, if army $X_i$ is attacking army $X_j$, then also $X_j$ is attacking $X_i$.} and is telling us who is fighting against whom. The matrix $C$ is generally not symmetric, and it represents the ongoing conflicts and the distributions of the various armies. A generalization of this construction taking into account an heterogeneous distribution of forces, or unilateral conflicts, i.e. a more general matrix $C$, could be used to describe more general scenarios. We remark that the matrix $\tilde{C}$ can be interpreted as the adjacency matrix of a graph with $n$ nodes (the armies), whose edges represent existing conflicts. The irreducibility assumption on $\tilde{C}$ impedes a splitting of the ongoing conflicts into two or more sub-conflicts, as it coincides with assuming that the corresponding graph is connected.

Moreover, we introduce the diagonal matrix $D$, whose diagonal entries are the parameters $d_{i,i}=x_i$. The $n$-dimensional generalization of system \eqref{lanch_aim} is given by
\begin{equation}\label{n-lanch_aim}
    X'=-CD X.
\end{equation}

With the same notation, system \eqref{lanch_unaim} generalizes to 
\begin{equation}\label{n-lanch_unaim}
    X_i'=(-CD X)_i X_i, \quad i=1,2,\dots,n.
\end{equation}
Clearly, for $n=2$, we get
$$
C
= \tilde{C} = \left( \begin{matrix} 0 & 1 \\ 1 & 0 \end{matrix} \right),
$$ 
thus, system \eqref{n-lanch_aim} reduces to system \eqref{lanch_aim}, and system \eqref{n-lanch_unaim} reduces to system \eqref{lanch_unaim}. However, in general, neither \eqref{n-lanch_aim} nor \eqref{n-lanch_unaim} admit constants of motion, which were used in the original models to predict which side would win the conflict, based solely on $X(0)$ and $x$.

We remark that for systems \eqref{lanch_aim} and \eqref{n-lanch_aim}, $X_i(t)=0$ does not stop the evolution in time for the corresponding army. Indeed, if we let \eqref{lanch_aim} evolve past the moment in which e.g. $R=0$ and $G>0$, we would see $R\rightarrow -\infty$ and $G\rightarrow +\infty$, as in Figure \ref{fig:quadr}. This clearly represents a criticality of such a simple model, which is inherited by its $n$-dimensional generalization \eqref{n-lanch_aim}.
 
On the contrary, both in system \eqref{lanch_unaim} and \eqref{n-lanch_unaim}, $X_i(\bar{t})=0$ implies that, for any $t>\bar{t}$, $X_i(t)=0$. Of course, the losing army in \eqref{lanch_unaim} only converges to $0$ as $t\rightarrow +\infty$. However, if we set the initial condition to one of the two to $0$, and the other to any positive constant, the system is at equilibrium, as we would expect from the interpretation of a conflict with only one army involved. The same holds true in the $n$-dimensional case: an army starting with $0$ soldiers will not be affected by the conflict between the remaining armies, as is biologically reasonable to happen.

In the following, we distinguish three cases: the \emph{free-for-all} case, in which the matrix $\tilde{C}$ has all 1 entries except for the zeros on the diagonal; the \emph{army-of-one} case, in which the matrix $\tilde{C}$ has 1 entries only in the first row and column, except the entry $(1,1)$, and every other case. We will show that, in the free-for-all scenario, only one winner will emerge from the conflict; in the army-of-one case, there can be up to $n-1$ winners (interpreted as asymptotically surviving armies); the outcome of the remaining cases is highly case-dependent.

We remark that, although we were able to recover a constant of motion for some particular cases, it is our opinion that for higher dimensions, i.e. high number of armies involved, the existence of such a constant is not guaranteed.

\subsection{Free-for-all scenario}
We now look at equations \eqref{n-lanch_aim} and \eqref{n-lanch_unaim} for the case $n=3$. Since we are in the free-for-all scenario, each army is facing all the other armies, and we have that
\begin{equation}\label{C3d}
C =\dfrac{1}{2}\begin{pmatrix}
0 & 1 & 1\\
1 & 0 & 1\\
1 & 1 & 0
\end{pmatrix}.
\end{equation}
Hence system \eqref{n-lanch_aim} becomes
\begin{equation}
\label{3-lanch_aim}
\begin{cases}
    X_1'&= \frac{1}{2}(-x_2 X_2 -x_3 X_3),\\
    X_2'&= \frac{1}{2}(-x_1 X_1 -x_3 X_3),\\
    X_3'&= \frac{1}{2}(-x_1 X_1 -x_2 X_2),
\end{cases}
\end{equation}
whose only equilibrium is $(0,0,0)$. System \eqref{3-lanch_aim} inherits the same critical issue of system \eqref{lanch_aim}: one (see Figure \ref{fig:neg1}) or two armies (see Figure \ref{fig:neg2}) can become negative, while the remaining diverge to $+\infty$.  

\begin{figure}[!ht]  
\centering
    \begin{subfigure}[t]{0.45\textwidth}
        \centering
    \includegraphics[width=0.99\textwidth] {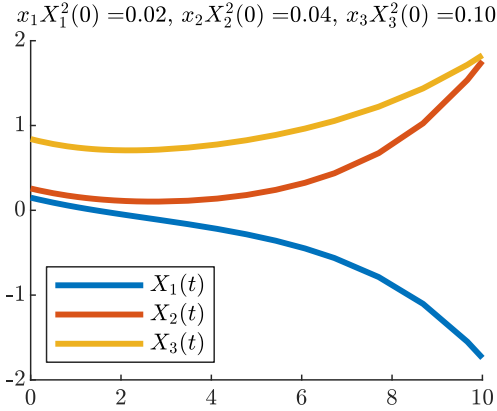}
    \caption{One negative, two positive diverging solutions.}
    \label{fig:neg1}
     \end{subfigure}
    \hspace{0.05\textwidth}
         \begin{subfigure}[t]{0.45\textwidth}
        \centering
         \centering
    \includegraphics[width=0.99\textwidth] {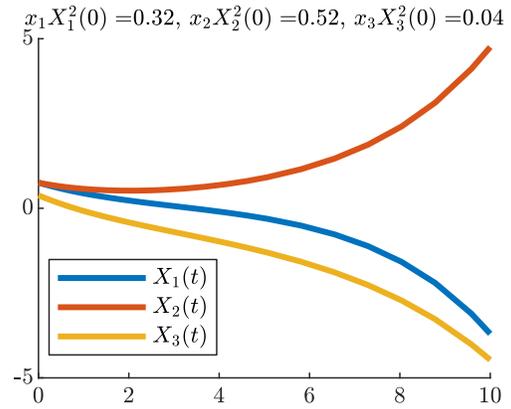}
    \caption{Two negative, one positive diverging solutions.}
    \label{fig:neg2}
         \end{subfigure}
         \caption{Evolution in time of the solutions of \eqref{3-lanch_unaim}, exhibiting (a) one negative and two positive diverging solutions, and (b) two negative and one positive diverging solutions. This fundamentally unrealistic behaviour can not be avoided with a simple linear system. The values relative to the Quadratic Law \eqref{quad_law} are displayed as the title of the figure, although they do not give useful information on the asymptotic behaviour of the middle solution.}
\end{figure}

In the remainder of this article, in order to overcome the critical issue of solutions diverging to $\pm\infty$, we will only focus on properties of generalized versions of \eqref{lanch_unaim}. In the following, we completely characterize the asymptotic behaviour of its 3D generalization. 
Using \eqref{C3d}, system \eqref{n-lanch_unaim} becomes
\begin{equation}
\label{3-lanch_unaim}
\begin{cases}
    X_1'&= \frac{1}{2}(-x_2 X_2 -x_3 X_3)X_1,\\
    X_2'&= \frac{1}{2}(-x_1 X_1 -x_3 X_3)X_2,\\
    X_3'&= \frac{1}{2}(-x_1 X_1 -x_2 X_2)X_3.
    \end{cases}
\end{equation}
Any point of the form $(X_1^*,0,0)$, $(0,X_2^*,0)$ or $(0,0,X_3^*)$ is an equilibrium of \eqref{3-lanch_unaim}. However, analysis of local stability at any of them provides one zero eigenvalue and two equal non-positive ones, namely $-x_iX_i^*/2$; hence, local stability analysis does not provide useful information. We will now directly proceed with a global stability analysis.

Since \eqref{3-lanch_unaim} is non-increasing in all its variables, and bounded (assuming non-negative initial conditions), each $X_i(t)$ will admit a finite non-negative limit as\linebreak $t \rightarrow +\infty$.

We are interested in predicting which army will win, given only the initial conditions $(X_1(0),X_2(0),X_3(0)) \in \mathbb{R}^3_{\geq 0}$ and the fighting abilities $(x_1,x_2,x_3) \in \mathbb{R}^3_{> 0}$.  We will focus on the case $(X_1(0),X_2(0),X_3(0)) \in \mathbb{R}^3_{> 0}$. Indeed if the initial condition is $X_i(0)=0$  for one $i$, system \eqref{3-lanch_unaim} reduces to \eqref{lanch_unaim}; if $X_i(0)=X_j(0)=0$ for $i\neq j$, the system is at equilibrium, and the same in the case $X_1(0)=X_2(0)=X_3(0)=0$, as one would expect.

\begin{proposition}\label{winner}
Assume that in system \eqref{3-lanch_unaim} initial conditions and parameters are such that $x_iX_i(0)>x_jX_j(0)\geq x_kX_k(0)$, for some permutation of $i,j,k\in\{1,2,3\}$. Then
$$
\lim_{t\rightarrow+\infty}X_i(t)>0 \quad \text{and} \quad \lim_{t\rightarrow+\infty}X_j(t)=\lim_{t\rightarrow+\infty}X_k(t)=0.
$$
\end{proposition}
\begin{proof}
Let us suppose, without loss of generality, that $x_1X_1(0)>x_2X_2(0)\geq x_3X_3(0)$. The proof for the remaining five cases is identical.

First of all, we remark that
$$
X_1'=\frac{1}{2}(-x_2 X_2 -x_3 X_3)X_1\geq \frac{1}{2}(-x_2 X_2(0) -x_3 X_3(0))X_1,
$$
since all the solutions are non-increasing, meaning
$$
X_1(t)\geq X_1(0) \text{exp}\bigg((-x_2 X_2(0) -x_3 X_3(0))\frac{t}{2}\bigg).
$$
The same reasoning applied to $X_2(t)$ and $X_3(t)$ ensure the solutions of \eqref{3-lanch_unaim} will never become negative.

Moreover, consider the first two equations of \eqref{3-lanch_unaim}; they can be bounded as follows:
$$
\begin{cases}
    X_1'&= \frac{1}{2}(-x_2 X_2 -x_3 X_3)X_1\leq -\frac{1}{2}x_2 X_1X_2,\\
    X_2'&= \frac{1}{2}(-x_1 X_1 -x_3 X_3)X_2\leq -\frac{1}{2}x_1 X_1X_2.
    \end{cases}
$$
This means that the trajectories of $X_1(t)$ and $X_2(t)$ are bounded from above by the respective solution of system 
$$
\begin{cases}
    Y_1'&= -\frac{1}{2}x_2 Y_1Y_2,\\
    Y_2'&= -\frac{1}{2}x_1 Y_1Y_2,
\end{cases}
$$
which is the classical 2D Lanchester model \eqref{lanch_unaim} with parameters $x_1/2$, $x_2/2$. From our assumptions, we know that $Y_1(t)\rightarrow Y_1(0)-x_2 Y_2(0)/x_1>0$, while $Y_2(t) \rightarrow 0$ as $t\rightarrow+\infty$, implying that $X_2(t) \rightarrow 0$, as well.

The same reasoning applied to the couple $X_1$ and $X_3$ shows that $X_3(t) \rightarrow 0$ as $t\rightarrow+\infty$.

We conclude by remarking that 
$$
X_1(t)=X_1(0)\text{exp}\bigg(-\dfrac{x_2}{2}\int_0^t X_2(s)\text{d}s -\dfrac{x_3}{2}\int_0^t X_3(s)\text{d}s\bigg).
$$
Recall \eqref{sol_unaim}: since both integrands in the exponential are bounded from above by exponentially vanishing functions, then
\begin{equation}\label{X1limit}
\lim_{t\rightarrow+\infty} X_1(t)=X_1(0)\text{exp}\bigg(-\dfrac{x_2}{2}\int_0^{+\infty} X_2(s)\text{d}s -\dfrac{x_3}{2}\int_0^{+\infty} X_3(s)\text{d}s\bigg)>0,
\end{equation}
which concludes the proof.
\end{proof}
\begin{remark} For our purpose, it is sufficient to know that the limit \eqref{X1limit} is strictly positive. However, one might obtain a lower bound for such a limit by bounding $X_2$ and $X_3$ from above with the explicit solution of the 2D Lanchester given in \eqref{sol_unaim}.
\end{remark}

The considerations above, and the proof of Proposition \ref{winner}, allow us to formulate the following, more general result.
\begin{theorem}\label{free-for-all}
Assume that the matrix $\tilde{C}$ 
is given as in the free-for-all scenario, i.e. every entry is 1 except the entries in the diagonal, which are 0. Then, the winner is the $i_W$-th army, where $i_W=\text{\emph{argmax} } x_i X_i(0)$.
\end{theorem}

\subsection{Army-of-one scenario}

The army-of-one scenario is, in a way, opposite to the free-for-all: $n-1$ armies attack the remaining one, which we will assume, without loss of generality, to always be $X_1$. It is natural to interpret $X_2$, $X_3$, $\dots$, $X_n$ as allies, since they only focus on a common enemy. A similar scenario is also studied by \cite{Lin2014}. Hence, there are two possible outcomes: either $X_1$ survives, killing all the other armies, or $X_1$ eventually perishes, and the remaining armies are victorious. 
\smallskip

The $n\times n$ matrices $\tilde{C}$ and $C$ in this case are respectively

\begin{equation*}
\tilde{C}=\left( 
\begin{NiceArray}{c|cw{c}{1cm}c}
0 & 1 & \dots & 1 \\
\hline
1 & \Block{3-3}<\Large>{\mathbf{0}}\\
\vdots & & & \\
1 & & & \\
\end{NiceArray}
\right)
\ \ \text{and}\ \ 
C=\left(
\begin{NiceArray}{c|cw{c}{1cm}c}
0 & 1 & \dots & 1 \\
\hline
\frac{1}{n-1} & \Block{3-3}<\Large>{\mathbf{0}}\\
\vdots & & & \\
\frac{1}{n-1} & & & \\
\end{NiceArray}
\right).
\end{equation*}

For $n=3$, the system of ODEs is the following:
\begin{equation}\label{2vs1_unaim}
\begin{cases}
    X_1'&= (-x_2 X_2 -x_3 X_3)X_1,\\
    X_2'&= \frac{1}{2}(-x_1 X_1)X_2,\\
    X_3'&= \frac{1}{2}(-x_1 X_1)X_3.
\end{cases}
\end{equation}
\begin{figure}  
\centering
    \begin{subfigure}[t]{0.45\textwidth}
        \centering
    \includegraphics[width=0.99\textwidth] {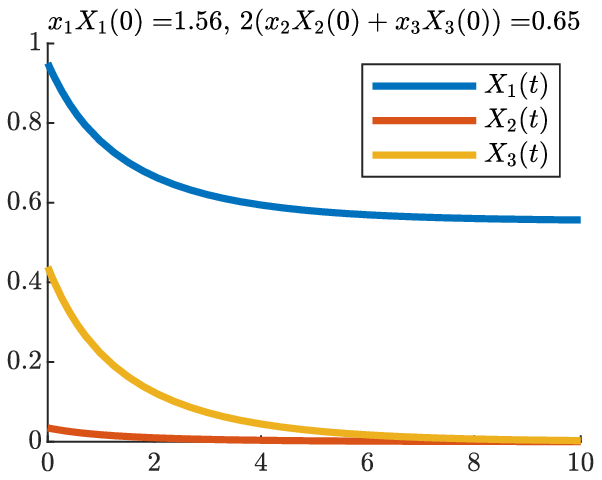}
    \caption{Army-of-one 3D, $X_1$ winning.}
    \label{fig:neg11}
     \end{subfigure}
    \hspace{0.05\textwidth}
         \begin{subfigure}[t]{0.45\textwidth}
        \centering
         \centering
    \includegraphics[width=0.99\textwidth] {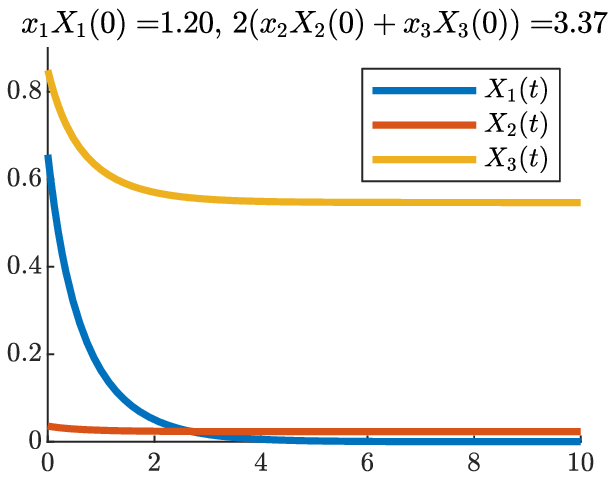}
    \caption{Army-of-one 3D, $X_1$ losing.}
    \label{fig:neg22}
         \end{subfigure}
         \caption{Evolution in time of the solutions of \eqref{2vs1_unaim}, exhibiting (a) $X_1$ surviving and $X_2$, $X_3$ vanishing, as $K>0$, and (b) $X_1$ vanishing and $X_2$, $X_3$ surviving, as $K>0$. Notice how $X_3$ in figure (b) does not vanish, even though it starts at very small value, due to its belonging to the winning side.}
\end{figure}

In this case, it is possible to recover a constant of motion which is closely related to the Lanchester Linear Law. In fact, by direct derivation with respect to the time variable $t$,
it is easy to see that the following holds, along solutions of \eqref{2vs1_unaim}:
\begin{equation}
\label{eq:2vs1}
    x_1X_1(t)-2\left ( x_2X_2(t)+ x_3X_3 (t)\right )= K,
\end{equation}
with $K$ a constant depending on $x_i$ and $X_i(0)$, $i=1,2,3$. In particular, as we will show in a more general setting in Theorem \ref{army-of-one}, \eqref{eq:2vs1} implies that, assuming the constant is positive at time $t=0$, $X_1$ will be the only survivor as $t\rightarrow +\infty$, see Figure \ref{fig:neg11}; whereas if the constant is negative, $X_1$ will perish and all the remaining forces will survive, see Figure \ref{fig:neg22}.

By using an analogous argument, one can easily prove a similar result for the general $n$-dimensional system:
\begin{equation}\label{Nvs1_unaim}
\begin{cases}
    X_1'&= (-x_2 X_2 -x_3X_3 -\dots -x_n X_n)X_1,\\
    X_2'&= \frac{1}{n-1}(-x_1 X_1)X_2,\\
     X_3'&= \frac{1}{n-1}(-x_1 X_1)X_3,\\
    \ \vdots\\
    X_n'&= \frac{1}{n-1}(-x_1 X_1)X_n.
\end{cases}
\end{equation}
The following statement summarizes the foregoing.

\begin{proposition}[Lanchester's $n$-dimensional Linear Law]
\label{Thm:ndim_LanchLinear}
The solutions to \eqref{Nvs1_unaim} satisfy the state equation
\begin{equation}\label{ndim_lin_law}
x_1X_1(t)-(n-1)(x_2X_2(t)+\dots + x_nX_n(t)) = K, \quad n\ge 2,
\end{equation}
where $K$ is a constant that depends only on the fighting abilities of the armies, $x_1, \dots x_n$, and on their initial size $X_1(0),\dots, X_n(0)$. 
\end{proposition}

The claim once again follows by deriving \eqref{ndim_lin_law} with respect to the time variable $t$ along the solutions of system \eqref{Nvs1_unaim}.

The sign of the constant of motion \eqref{ndim_lin_law} allows us to predict the outcome of the conflict, as stated in the following theorem.

\begin{theorem}\label{army-of-one}
Recall state equation \eqref{ndim_lin_law} and system \eqref{Nvs1_unaim}; if at $t=0$ the sign of the constant $K$ is positive (resp. negative), $X_1$ will win the war (resp. $X_1$ will perish).
\end{theorem}

\begin{proof}
From the arguments used in the proof of Proposition \ref{winner}, we know that all $X_i$, for $i=1,2,\dots,n$ will remain non-negative.

We introduce the auxiliary variable
$$
Z:=x_2X_2+\dots+x_nX_n.
$$
Then, the ODE governing the evolution of the $(X_1,Z)$ couple is
$$
\begin{cases}
    X_1'&= -ZX_1,\\
    Z' &= -\frac{x_1}{n-1} X_1Z.
\end{cases}
$$

Now, assume $K>0$. Then,
$$
\frac{x_1}{n-1}X_1(0)>Z(0),
$$
implies the asymptotic extinction of $Z$, which bounds each $X_i$ to the same fate. Since $Z(t)$ is decaying exponentially fast (recall Lemma \ref{expo_dec}), we know that 
$$
\lim_{t\rightarrow+\infty} X_1(t)=X_1(0) \text{exp}\bigg(-\int_0^{+\infty} Z(s)\text{d}s\bigg)>0. 
$$

Now, assume $K<0$. 
Then,
$$
\frac{x_1}{n-1}X_1(0)<Z(0),
$$
implies the exponential vanishing of $X_1$.
Lastly, observing that
$$
\lim_{t\rightarrow +\infty} X_i=X_i(0) \text{exp}\bigg(-\dfrac{x_1}{n-1}\int_0^{+\infty} X_1(s)\text{d}s\bigg)>0,
$$
since $X_1(t)$ is decaying exponentially fast (recall again Lemma \ref{expo_dec}), we know that all the allied armies will survive the conflict. 
\end{proof}

\begin{remark} Modelling combat with a continuous approach can, in some cases, provide counter intuitive results, such as very small armies in \eqref{Nvs1_unaim} not vanishing since they belong to the winning side of the conflict. In a discrete time model, a very small army might vanish entirely even though its allies win the war, as explained, for instance, by \cite{Fox2010}.
\end{remark}

\begin{remark} We remark that in the case $n=3$, the free-for-all and army-of-one cases completely describe every possible scenario, up to a renumbering of the armies involved.
\end{remark}

\begin{remark}
\label{Rmk:Lepin}
John \cite{Lepingwell1987} proposed a model, in which one of the two armies is split in two different units. The model is similar to the aimed version of \eqref{2vs1_unaim}. However, his model takes into account the proportional distribution of the forces of the single army between the two units of the other two forces (i.e. the single army distributes the attacks in proportion to the number or survivors of the other armies). Even thought the idea is brilliant, this model is affected by the same problems as the original one, namely solutions eventually assuming negative values.
\end{remark}

\subsection{Four armies example and general result}

Lastly, we provide an interesting example involving four armies, which does not fall in the free-for-all nor in the army-of-one setting.
We remark that the only cases for which, thus far, we have characterized the asymptotic behaviour are the $n$-dimensional free-for-all, see Theorem \ref{free-for-all}, and the $n$-dimensional army-of-one, see Theorem \ref{army-of-one}. 

However, for $n\geq 4$, there can be scenarios which do not fall in either of the two cases described thus far; see Figure \ref{Figura4}. The review of \cite{MacKay2006} contains various examples of multilateral conflicts, in particular the case of mixed forces, in which two or more armies fight from various regions. Moreover, we refer the interested reader to \cite{Kress2018b}. In the following, we describe a general scenario of multilateral conflict which, to our knowledge, has not been studied in the literature.

We introduce the following definition, for ease of notation.

\begin{definition}
\label{def:subconflict}
We define a \emph{conflict} as a triple $\kappa=(X(0),\tilde{C},D)$, where $X(0)$ is the vector $(X_1(0),\dots,X_n(0))$ of the initial values of the involved armies, $\tilde{C}$ is the adjacency matrix of Definition \ref{def_conflictmatrix}, and $D$ is the diagonal matrix such that $d_{i,i}=x_i$. 

A \emph{sub-conflict} $\kappa^j$ is given by a subset $X^j(0)$ of the armies and the corresponding matrices $\tilde{C}^j$ and $D^j$, stripped of the rows and columns of the armies not considered.
\end{definition}

As we already mentioned at the beginning of Section \ref{Sec3}, to a conflict $\kappa$ we can associate two systems of ODEs:
$$X' = -CDX,$$ which correspond to the aimed fire system in Lanchester's approach, and $$X'_i = (-CDX)_i X_i,$$ which correspond to the unaimed fire one. Here the matrix $C$ is derived from the matrix $\tilde{C}$ as in Definition \ref{def_conflictmatrix}.

\begin{definition}\label{def:bridge}
A node $X_i$ is called \emph{bridge node} if it divides the multilateral conflict $\kappa$ in $m\geq 2$ sub-conflicts $\kappa^1,\dots,\kappa^m$, such that 
$$
\bigcap_{j=1}^m X^j = X_i.
$$
\end{definition}

\tikzstyle{armata}=[rectangle,draw=black!50,fill=black!10,thick,minimum size=1cm]
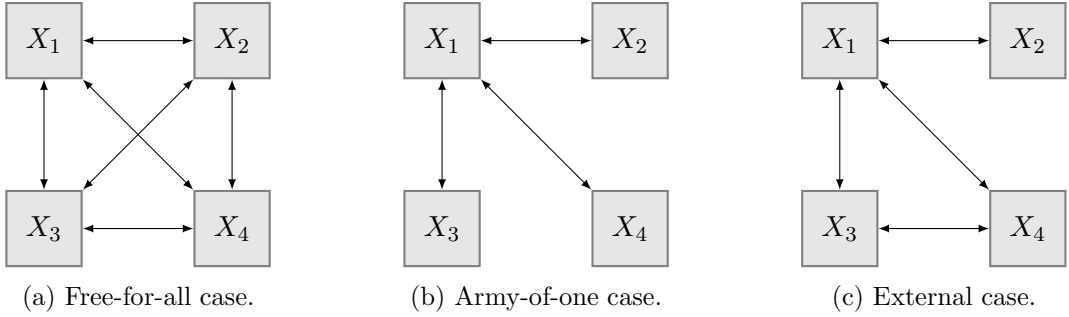
\begin{figure}[ht]
\centering
\begin{subfigure}{0.3\textwidth}\centering
\begin{tikzpicture}
\node[armata] (X_3) at ( 0,0) {$X_3$};
\node[armata] (X_1) at ( 0,2.5) {$X_1$};
\node[armata] (X_4) at ( 2.5,0) {$X_4$};
\node[armata] (X_2) at ( 2.5,2.5) {$X_2$};
\draw[latex-latex] (X_1.east) -- (X_2.west);
\draw[latex-latex] (X_3.east) -- (X_4.west);
\draw[latex-latex] (X_1.south) -- (X_3.north);
\draw[latex-latex] (X_4.north) -- (X_2.south);
\draw[latex-latex] (X_1.south east) -- (X_4.north west);
\draw[latex-latex] (X_3.north east) -- (X_2.south west);
\end{tikzpicture}
\caption{Free-for-all case.}
\label{Figura4a}
\end{subfigure}
\hspace{.5cm}
\begin{subfigure}{0.3\textwidth}\centering
\begin{tikzpicture}
\node[armata] (X_3) at ( 0,0) {$X_3$};
\node[armata] (X_1) at ( 0,2.5) {$X_1$};
\node[armata] (X_4) at ( 2.5,0) {$X_4$};
\node[armata] (X_2) at ( 2.5,2.5) {$X_2$};
\draw[latex-latex] (X_1.east) -- (X_2.west);
\draw[latex-latex] (X_1.south) -- (X_3.north);
\draw[latex-latex] (X_1.south east) -- (X_4.north west);
\end{tikzpicture}
\caption{Army-of-one case.}
\label{Figura4b}
\end{subfigure}
\hspace{.5cm}
\begin{subfigure}{0.3\textwidth}
\centering
\begin{tikzpicture}
\node[armata] (X_3) at ( 0,0) {$X_3$};
\node[armata] (X_1) at ( 0,2.5) {$X_1$};
\node[armata] (X_4) at ( 2.5,0) {$X_4$};
\node[armata] (X_2) at ( 2.5,2.5) {$X_2$};
\draw[latex-latex] (X_1.east) -- (X_2.west);
\draw[latex-latex] (X_3.east) -- (X_4.west);
\draw[latex-latex] (X_1.south) -- (X_3.north);
\draw[latex-latex] (X_1.south east) -- (X_4.north west);
\end{tikzpicture}
\caption{External case.}
\label{Figura4c}
\end{subfigure}
\caption{Visual representation of the 4D cases (a) free-for-all; (b) army-of-one; and (c) a case, neither free-for-all nor army-of-one, which we discuss in this section.}
\label{Figura4}
\end{figure}

We now analyse the scenario depicted in Figure \ref{Figura4c}. In our modelling approach it is described by the following system of ODEs:
\begin{equation}\label{4Dgen}
\begin{cases}
    X_1'&= (-x_2 X_2 -\frac{x_3}{2}X_3 -\frac{x_4}{2} X_4)X_1,\\
    X_2'&= -\frac{x_1}{3} X_1 X_2,\\
     X_3'&= (-\frac{x_1}{3} X_1-\frac{x_4}{2} X_4)X_3,\\
    X_4'&=(-\frac{x_1}{3} X_1 -\frac{x_3}{2}X_3)X_4.
\end{cases}
\end{equation}
From Definition \ref{def:bridge} it is easy to see that $X_1$ is  a bridge node for system \eqref{4Dgen}. Intuitively, one would be tempted to ``split'' the dynamics into two free-for-all subsystems (namely, $X_1$ vs. $X_2$ and $X_1$ vs. $X_3$ vs. $X_4$), and check the asymptotic behaviour for both of them with Theorem \ref{free-for-all}. The forces of $X_1$ would then be split in 3, one per each army it is facing.

This would mean finding the maxima in two sets: $\{x_1X_1(0)/3, x_2X_2(0)\}$ and \linebreak $\{x_1X_1(0)/3, x_3X_3(0)/2, x_4X_4(0)/2\}$, as highlighted in the titles of Figure \ref{fig:4D}. 

There are only two clear-cut cases: the first is when the maximum of both sets is $x_1X_1(0)/3$, as $X_1$ will result victorious against all its 3 opponents; see Figure \ref{fig:4D1}. The second is when the maximum of each set is $x_2X_2(0)$ and one among $x_3X_3(0)/4$ or $x_4X_4(0)/4$; see Figure \ref{fig:4D2}.

The remaining cases are not so simple to interpret at glance, but they can be described analytically. For example if $X_1$ would win against $X_3$ and $X_4$, but $X_2$ would win against $X_1$, then one should expect $X_2$ as the only winner. However, our simulations show that $X_1$ will eventually perish and the winner among $X_3$ and $X_4$ will be determined by their respective strength; see Figure \ref{fig:4D3}. Vice versa, if $X_1$ would win against $X_2$, but not in the free-for-all with $X_3$ and $X_4$, one would expect to see one between $X_3$ and $X_4$ as the only winner, with all the other armies reaching 0. However, in our simulations $X_1$ perishes (as expected) and $X_2$ survives; see Figure \ref{fig:4D4}.

Assume, for example, that $X_1$ would lose in the $X_1$ vs. $X_2$ sub-conflict, i.e. $x_1 X_1(0)/3< x_2 X_2(0)$, but would otherwise win in the $X_1$ vs. $X_3$ vs. $X_4$ one, i.e. $x_1 X_1(0)/3> x_3 X_3(0)/2, x_4 X_4(0)/2$. Then we can focus only on these two variables and bound their respective ODE from above by
$$
\begin{cases}
    X_1'&= (-x_2 X_2 -\frac{x_3}{2}X_3 -\frac{x_4}{2} X_4)X_1 \leq -x_2 X_2 X_1,\\
    X_2'&= -\frac{x_1}{3} X_1 X_2.
\end{cases}
$$
Recalling \eqref{sol_unaim}, we know that $X_1(t)$ is bounded from above by an exponentially vanishing function. Similarly, we can bound the ODEs governing $X_3$ and $X_4$ by
$$
\begin{cases}
    X_3'&= (-\frac{x_1}{3} X_1-\frac{x_4}{2} X_4)X_3 \leq -\frac{x_4}{2} X_4 X_3,\\
    X_4'&=(-\frac{x_1}{3} X_1 -\frac{x_3}{2}X_3)X_4\leq -\frac{x_3}{2}X_3 X_4.
\end{cases}
$$
Assume now, without loss of generality, that $x_3X_3(0)>x_4X_4(0)$. Then, again by \eqref{sol_unaim}, we know that $X_4(t)$ is bounded from above by an exponentially vanishing function. Combining these remarks, we obtain
$$
\lim_{t\rightarrow +\infty} X_3(t)= X_3(0) \text{exp}\bigg(-\dfrac{x_1}{3}\int_0^{+\infty} X_1(s)\text{d}s -\dfrac{x_4}{2}\int_0^{+\infty} X_4(s)\text{d}s \bigg)>0.
$$
The case in which $X_1$ would win against $X_2$, but would otherwise lose in the $X_1$ vs. $X_3$ vs. $X_4$ subconflict, can be analysed similarly.
\begin{figure}[ht!]  
\centering
    \begin{subfigure}[t]{0.45\textwidth}
        \centering
    \includegraphics[width=0.99\textwidth] {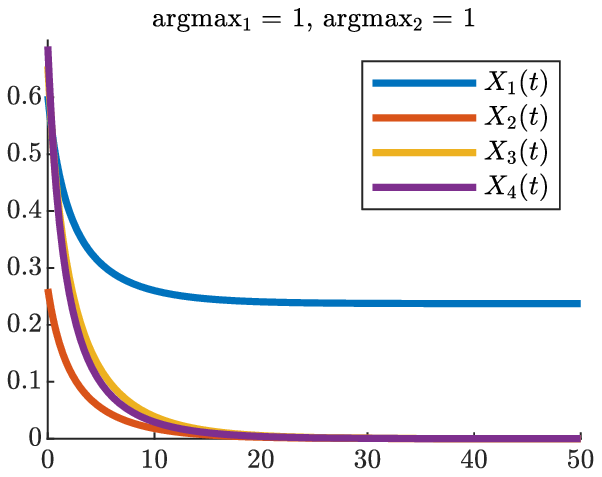}
    \caption{$X_1$ wins.}
    \label{fig:4D1}
     \end{subfigure}
    \hspace{0.05\textwidth}
         \begin{subfigure}[t]{0.45\textwidth}
        \centering
         \centering
    \includegraphics[width=0.99\textwidth] {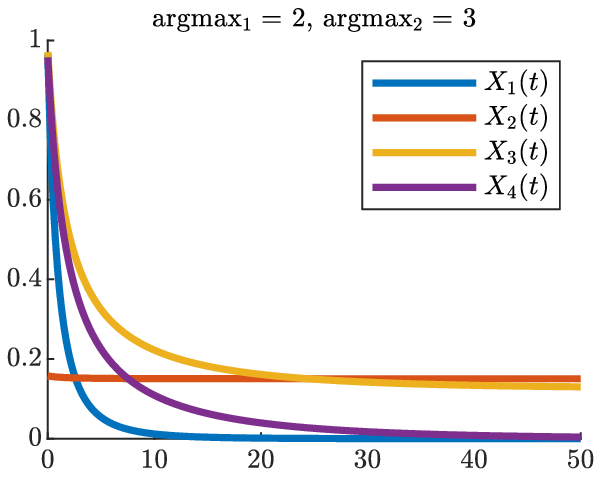}
    \caption{$X_2$ and $X_3$ win.}
    \label{fig:4D2}
         \end{subfigure}
                  \begin{subfigure}[t]{0.45\textwidth}
        \centering
         \centering
    \includegraphics[width=0.99\textwidth] {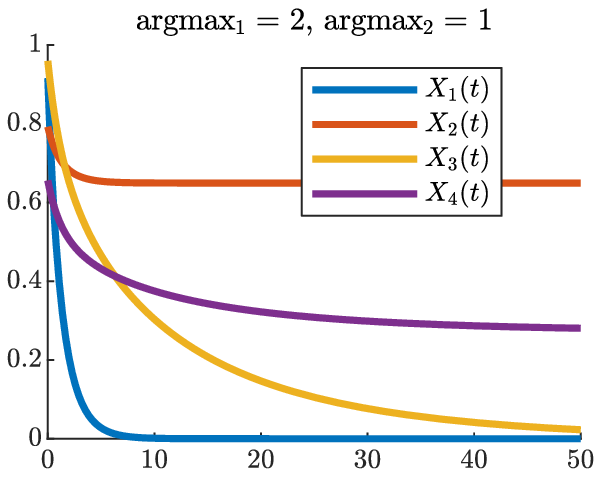}
    \caption{$X_4$ is the ``unexpected'' survivor.}
    \label{fig:4D3}
         \end{subfigure}
         \begin{subfigure}[t]{0.45\textwidth}
        \centering
         \centering
    \includegraphics[width=0.99\textwidth] {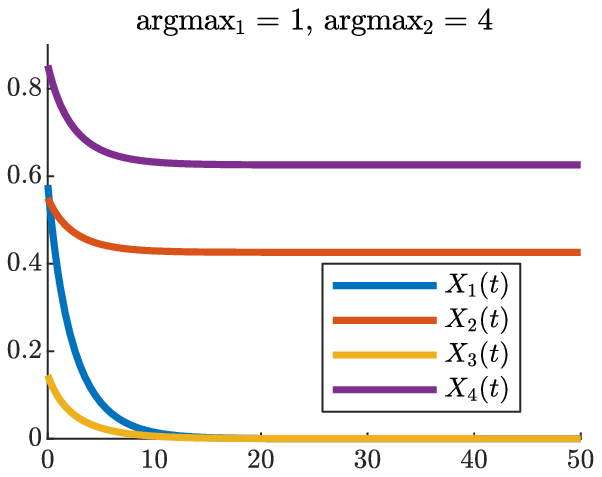}
    \caption{$X_2$ is the ``unexpected'' survivor.}
    \label{fig:4D4}
         \end{subfigure}
         
         \caption{Evolution in time of the solutions of \eqref{4Dgen}, in four different scenarios. As highlighted by the title, if we looked at the two conflicts separately, we would expect (a) $X_1$ to win both; (b) $X_2$ and $X_3$ to win; (c) $X_2$ to win against $X_1$, and $X_1$ to win against $X_3$, $X_4$; and (d) $X_1$ to win against $X_2$, and $X_4$ to win against $X_1$, $X_3$. Notice that only the first two predictions are satisfied.}\label{fig:4D}
\end{figure}

These considerations highlight the contribution of \emph{bridge} nodes, such as $X_1$ in our example, which play a crucial role in higher dimensional conflicts. Indeed, if they (asymptotically) vanish, the remaining conflicts become disconnected, and multiple sub-conflict proceed to the respective asymptotic outcome. We summarize the content of this section in the following statement.

\begin{theorem}\label{thm:bridge}
Assume $X_i$ is a bridge node for a conflict $\kappa$. Assume there exists a sub-conflict in which $X_i$ (asymptotically) loses. Then, $X_i$ loses in all other sub-conflicts as well. Moreover, each sub-conflict $\kappa^j$ in which $X_i$ is involved evolves qualitatively as $\kappa^j_i$, which is the sub-conflict of $\kappa^j$ without $X_i$. Hence, the winner of each of those sub-conflict is determined comparing the initial force of the remaining armies, i.e. by analyzing the sub-conflict $\kappa^j_i$.
\end{theorem}

We omit the general proof of Theorem \ref{thm:bridge}, since it would require a cumbersome notation and would follow the same steps of the discussion which precedes it.
\begin{remark}
The 4D case actually highlights a limitation of our model. Let us consider the example depicted in Figure \ref{fig:4D3}: even though $X_1$ quickly approaches $0$, $X_3$ and $X_4$ keep attaching each other with only \emph{half} of their respective forces. This can be interpreted as an initial distributions of forces on multiple fronts, which do not get rearranged even if a front becomes quiet, having defeated the corresponding opponent. However, if we were to replace the homogeneous distribution of forces with a distribution proportional to the remaining enemies in each army (by using the same idea described in Remark \ref{Rmk:Lepin}, for example), then we would get the problem we already discussed of some armies becoming negative.
\end{remark} 

\section{Conclusion}
\label{Sec4}
Lanchester's first models have paved the way to more advanced modelling efforts. This paper provides a new approach, which we believe grants more in-depth explorations. Our formalization may sometimes provide counter intuitive results, which could hopefully be avoided with a discretization of the problem. In this section, we propose various outlooks, which represent the first step of a wider, mathematically consistent body of research.

The research carried out for this paper opened many interesting directions for extending Lanchester's model of warfare. 

From a mathematical point of view, the introduction of the randomness in the model (in terms of stochastic processes) is certainly worthwhile. Some advancements in this direction have already been made, for instance, by  \cite{Kingman2002} and by \cite{Osborne2003}. Moreover, the link with game theory, as well as decision theory, is clear and it is well exposed by \cite{Chen2011}, \cite{JJ2011}, and \cite{Chen2012}. This perspective connects in particular to the division of forces in a multilateral conflict, which in this first paper we assumed to be homogeneous. In a general scenario, which is the best, possibly heterogeneous and time-varying, division of forces for a given army?

As explained in Section \ref{Sec3}, a discrete modeling constitutes a stimulating challenge, which may lead to a more realistic idealization of the warfare problem. 

Finally, we remark the possibility of developing a new unconventional approach based on graph theory: this view could lay out a more in-depth exploration of higher dimensional cases, depending in particular on the adjacency matrix.

For all this reasons, we believe Lanchester's type models pose an attractive topic of research in the context of applied mathematics, and the study of the issues arising from their analysis can involve many other areas of mathematics, making this subject still fascinating.
\bigskip

\par\bigskip\noindent
\textbf{Acknowledgments.} 
NC, MC, and MS would like to thank the Polytechnic University of Milan, Polish Academy of Science, and 
INRIA, respectively, for supporting their research.

\bibliographystyle{apalike}  \small
\bibliography{biblio}
\end{document}